\documentclass[ 10 pt, conference]{ieeeconf}  

\IEEEoverridecommandlockouts                              

\usepackage{comment}

\usepackage{amsmath,graphicx,amsfonts,amssymb,epsfig,subfig,mathrsfs,mathtools}
\usepackage{cuted}
\usepackage{arydshln}
\usepackage{blkarray}

\usepackage{color}
\usepackage{multirow}
\usepackage{multicol}
\usepackage{lipsum}
\usepackage{rotating}
\usepackage{graphicx}%
\usepackage{algorithm}
\usepackage{algpseudocode}
\usepackage{cite}
\usepackage{tikz}
\usepackage{cancel}
\usepackage{textcomp}

\DeclareMathOperator*{\minimize}{minimize}

\DeclareMathOperator*{\argmin}{arg\,min}

\newtheorem{assumption}{Assumption}
\newtheorem{theorem}{Theorem}

\newtheorem{proposition}{Proposition}

\allowdisplaybreaks

\newlength\myindent
\def\BibTeX{{\rm B\kern-.05em{\sc i\kern-.025em b}\kern-.08em
    T\kern-.1667em\lower.7ex\hbox{E}\kern-.125emX}}

\title{\LARGE \bf
Efficient, Decentralized, and Collaborative Multi-Robot Exploration using Optimal Transport Theory
}

\author{Rabiul Hasan Kabir$^{1}$ and Kooktae Lee$^{1}$
\thanks{$^{1}$R. H. Kabir and K. Lee are with the Department of Mechanical Engineering, New Mexico Institute of Mining and Technology, Socorro, NM 87801, USA.
        {\tt\small rabiul.kabir@student.nmt.edu, kooktae.lee@nmt.edu}}%
}

\begin{document}

\maketitle
\thispagestyle{empty}
\pagestyle{empty}

\begin{abstract}
An Optimal Transport (OT)-based decentralized collaborative multi-robot exploration strategy is proposed in this paper. This method is to achieve an efficient exploration with a predefined priority in the given domain. In this context, the efficiency indicates how a team of robots (agents) cover the domain reflecting the corresponding priority map (or degrees of importance) in the domain. The decentralized exploration implies that each agent carries out their exploration task independently in the absence of any supervisory agent/computer. When an agent encounters another agent within a communication range, each agent receives the information about which areas are already covered by other agents, yielding a collaborative exploration. The OT theory is employed to quantify the difference between the distribution formed by the robot trajectories and the given reference spatial distribution indicating the priority. A computationally feasible way is developed to measure the performance of the proposed exploration scheme. Further, the formal algorithm is provided for the efficient, decentralized, and collaborative exploration plan.
Simulation results are presented to validate the proposed methods.
\end{abstract}

\section{\uppercase{Introduction}}
A multi-robot exploration problem has been both widely and deeply investigated for more than decades due to the obvious reasons - less prone to failure than a single-robot system as well as time reduction to cover a given domain. Although there exist numerous research works related to multi-robot explorations, it can be categorized into three different fields -- Coverage Path Planning, Multi-Robot Exploration and Search, and Ergodic Exploration.

Coverage Path Planning (CPP) refers to a method to synthesize a robot path for passing over all points of an area or volume of interest. Some previous works for the multi-robot CPP problem include multi-robot lawnmower \cite{azpurua2018multi}, cell decomposition technique \cite{xu2014efficient}, \cite{avellar2015multi}, spanning tree-based CPP \cite{kim2014time}, Vornoi Diagram method \cite{yazici2013dynamic}, \cite{adaldo2017cooperative}. Incremental random planners such as Rapidly exploring Random Trees (RRT) and Probabilistic Road Map (PRM) are also in the category of CPP, which has very broad research works.

Multi-Robot Exploration and Search is for either finding a moving target in an indoor environment based on the Bayesian measurement update model \cite{hollinger2009efficient}, \cite{best2018online} or searching targets using Particle Swarm Optimization (PSO)-based approaches \cite{pugh2007inspiring}, \cite{kwa2020optimal}.

All previous works mentioned above, however, have only focused on the entire coverage of the given domain while not taking into account relative importance or priority of areas of interest, making the existing methods far from efficient exploration.

In \cite{GM-IM:11}, Mathew and Mezi{\'c} addressed a multi-robot exploration problem based on the ergodicity. In general, the ergodicity refers to system characteristics such that the time-averaged dynamics are equal to the given spatial average. In this work, a metric is defined to measure the ergodicity as the difference between the time-averaged multi-robot trajectory and the given spatial distribution. The Fourier basis functions are employed to facilitate the derivation of the ergodic control laws. This method has been further investigated and applied to many other works including \cite{silverman2013optimal},
\cite{lee2018receding},  
\cite{veitch2019ergodic}.

All of these works rely on the proposed result in \cite{GM-IM:11}, yet it contains the following issues.
The proposed result is developed for the centralized control scheme, which may not be desirable in practice. A computational issue arises in the implementation stage due to infinite numbers of the Fourier basis functions being used in the method. 
Finally, and most importantly, the ergodicity can be only achieved with infinite time, which is the fundamental limitation of the ergodic approach. This problem is fatal as robots have finite energy and hence, the ergodicity will never be attained in reality.

In this paper, we propose an efficient, decentralized, and collaborative multi-robot exploration scheme based on the optimal transport theory. To quantify the difference between the distribution obtained from multi-robot trajectories and a given spatial distribution, the OT theory is employed.
In \cite{kabir2020receding}, a preliminary result was introduced for an efficient single-robot exploration plan. This work has laid the foundation and opened up the possibility to generate an efficient robot trajectory based on the OT theory. This preliminary work, however, was developed for a single robot and did not consider the majority of research works investigated in this paper such as multi-robot trajectory generation, non-overlapping issues between multiple robots, and a decentralized control scheme. 

The major contributions of this paper are as follows. Firstly, an efficient multi-robot exploration plan is proposed to reflect priority of areas in the domain, given as a spatial reference distribution. Secondly, the proposed method is developed for decentralized exploration, which is more practical than a centralized scheme. Thirdly, through the proposed plan a multi-robot system is able to collaboratively complete an exploration mission, resulting in a faster coverage compared to a single-agent scenario.
Lastly, an upper bound of the performance measure for the exploration efficiency is derived. This bound can be calculated in a computationally efficient manner.
To validate the proposed method, simulation results are provided.


\textit{Notation:} A set of real and natural numbers are denoted by $\mathbb{R}$ and $\mathbb{N}$, respectively. Further, $\mathbb{N}_0 = \mathbb{N}\cup\{0\}$. The symbols $\Vert \cdot \Vert$ and $^{T}$, respectively, denote the Euclidean norm and the transpose operator. The symbol $\mathcal{R}(x,r)$ represents a set of points within the circle centered at $x$ with a radius $r$. The symbol $\#$ indicates the cardinality of a given set. The variable $t\in\mathbb{N}_0$ is used to denote a discrete time.

\section{\uppercase{Preliminary and Problem Description}}

In many practical scenarios, a domain associated with  different degrees of priority is necessary for efficient explorations. In this case, a team of robot can be deployed to explore the domain such that they investigate high-priority regions more frequently while spend less time in low-priority regions.

This study utilizes the OT theory as a tool to achieve this goal.
Traditionally, the optimal transport is to seek an optimal solution for a resource allocation problem \cite{villani2008optimal}. Among many different problem formulations based on the OT theory,
the Wasserstein distance \cite{villani2008optimal} of order $p$ is introduced as follows.
\begin{itemize}
\item Wasserstein distance:
\begin{align*}
W_p(\mu,\nu) := \left(\inf_{\gamma\in\Gamma(\mu,\nu)}\int_{X\times Y} \Vert x - y\Vert^p d\gamma(x,y)\right)^{\frac{1}{p}},
\end{align*}
\end{itemize}
The Wasserstein distance describes the least amount of effort to transform one distribution $\mu$ into another one $\nu$. This Wasserstein distance has been employed to broad dynamical systems including system analysis \cite{lee2014probabilistic}, \cite{lee2015performance}, \cite{lee2015analysis} as well as controller synthesis \cite{lee2014optimal}, \cite{lee2018optimal} problems.

In the discrete marginal case, the Hitchcock-Koopmans transportation problem \cite{evans1997partial} is developed for the optimal transport problem, where $\mu$ and $\nu$ are represented by particles.
The following linear programming (LP) formulation of the transportation problem is equivalent to the Wasserstein distance in the sample point representation of given distributions.
\begin{itemize}
\item Linear Programming problem: (for $p=1$) 
\begin{equation}\label{eqn: LP}
  \begin{aligned}
    & \underset{\pi_{ij}}{\text{minimize}} & & \sum_{i,j}\pi_{ij}\Vert x_i-y_j \Vert\\
    & \text{subject to} & & \pi_{ij} \geq 0,\\
   	& & & \sum_{j=1}^{N}\pi_{ij} = m(x_i), i=1,2,\ldots,M,\\
	& & & \sum_{i=1}^{M}\pi_{ij} = n(y_j), j=1,2,\ldots,N,
\end{aligned}
\end{equation}
\end{itemize}
where $\{x_i\}_{i=1}^{M}$ and $\{y_j\}_{j=1}^{N}$ are the set of sample points for two ensembles, $M$ and $N$ are respectively the number of sample points for $\{x_i\}$ and $\{y_j\}$, $m(x_i)$ and $n(y_j)$ are some non-negative constants representing the mass or weight corresponding to each particle in the ensemble. The variable $\pi_{ij}$ denotes the transport plan indicating how much mass transportation is required from $x_i$ to $y_j$.
The optimal transport plan $\pi_{ij}^*$ aims to determine an optimal solution for the minimum effort that is necessary to transport the weights. 

\begin{figure*}[t]
\begin{center}
\subfloat[given spatial distribution]{\includegraphics[scale=0.4]{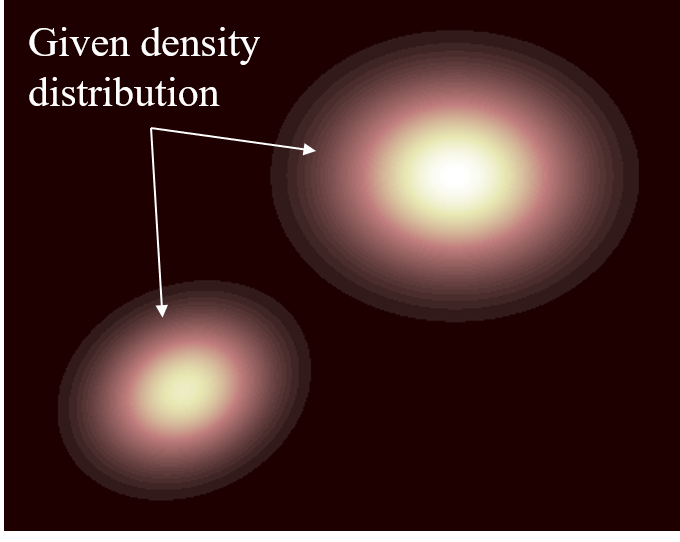}}\qquad
\subfloat[sampling]{\includegraphics[scale=0.4]{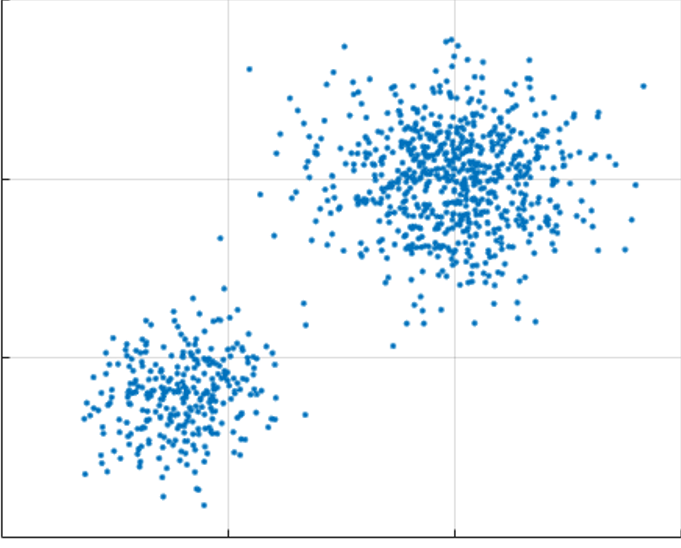}}\qquad
\subfloat[efficient robot exploration ]{\includegraphics[scale=0.4]{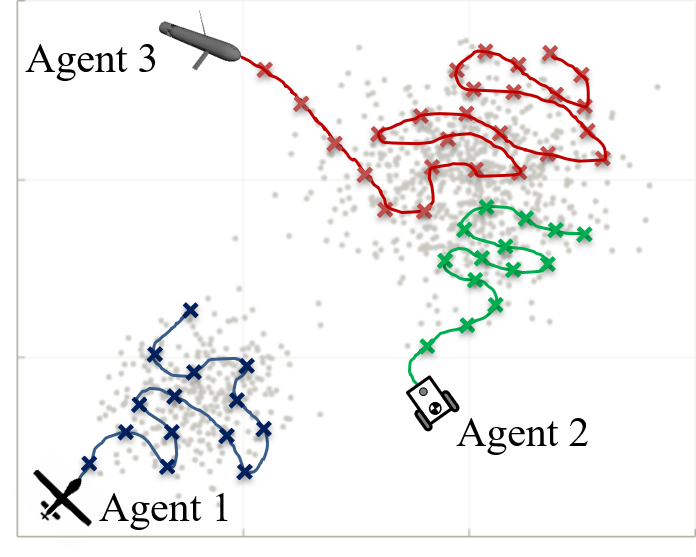}}
\end{center}
\caption{The procedure to generate the efficient, decentralized, and collaborative multi-robot trajectory using the OT theory}\label{fig: ergodic traj procedure}
\end{figure*}

In the decentralized multi-robot trajectory generation problem, the set of robot points $\{x_i^k\}$, where $k\in\{1,2,\ldots,n_a\}$ is the agent number and $n_a$ is the total number of agents, are not predetermined and hence, one needs to develop a strategy on how to obtain $\{x_i^k\}$. The Wasserstein distance in the LP form \eqref{eqn: LP} will be employed as a tool to measure the difference between the two ensembles, one from the robot trajectories, $\{x_i^k\}$, and another from the given reference distribution, $\{y_i\}$.
Therefore, the major goal of this research is to plan  the multi-robot trajectories such that the set $\{x_i\} := \bigcup_{k=1}^{n_a}\{x_i^k\}$ gets close to $\{y_i\}$, resulting in the efficient multi-robot exploration.
Mathematically, it is equivalent to generate robot trajectories $\{x_i\}$ such that $\displaystyle \sum_{i=1}\sum_{j=1}\pi_{ij}\Vert x_i - y_j\Vert \rightarrow 0$, with the given constraints in \eqref{eqn: LP}. 

Fig. \ref{fig: ergodic traj procedure} illustrates the schematic of the problem. The given spatial distribution (Fig. \ref{fig: ergodic traj procedure} (a)), is transformed into a sample point representation (Fig. \ref{fig: ergodic traj procedure} (b)), followed by generating the multi-agent trajectories to match  $\{x_i\}$ and $\{y_j\}$ (Fig. \ref{fig: ergodic traj procedure} (c)).

One simple and naive way to achieve this goal is making $\{x_i\}$ identical to $\{y_j\}$. However, this approach is impractical due to the following reasons:
\begin{enumerate}
\item[1)] Due to the motion constraints, the robots may not visit the sample point $y_j$.
\item[2)] As there exists an energy limitation for each robot, the total number of robot points denoted by $M$ is also limited and finite. Hence, $M$ may be smaller than the total number of sample points given by $N$.
\item[3)] For $M\neq N$, it is not possible to match the robot points with the sample points.
\item[4)] Even for $M=N$, where $N$ is a very large number, it may take an excessive amount of time for robots to survey the domain while following the trajectory generated by connecting all the sample points sequentially. 
\end{enumerate}
In the sequel, the OT-based decentralized collaborative multi-agent exploration scheme is provided to ensure efficient exploration of a given domain while avoiding these issues stated above.
\section{\uppercase{OT-based Decentralized Multi-Agent Exploration}}\label{section 3}

This section provides a key idea for the efficient, decentralized, and collaborative multi-robot exploration based on the OT theory. 
It is easier to implement single agent exploration strategy to explore a domain due to the absence of complexities associated to multi-agent systems, for example, the coordination between agents, information sharing, collision avoidance between agents, etc. However, a single-agent exploration scheme is inefficient for applications with very spacious domain.
Therefore, to maximize the exploration efficiency, it is better to utilize a team of agents instead of a single agent.
The centralized case can be thought of as a multi-agent system with the assumption that there exists a supervisory agent that receives all relevant information from each subordinate agent and shares the information with all other agents, enabling all agent to realize a coordinated exploration plan. This scenario is effective and applicable only if there are no communication interruptions between all agents, which is very restrictive in practice.
Moreover, a centralized control strategy is more vulnerable to a single point of failure (i.e., a breakdown of the supervisory agent will lead to the failure of the whole system).

To avoid the aforementioned issues associated with the centralized control approach, the decentralized collaborative multi-agent exploration scheme is developed here while considering a limited communication range. The decentralized control implies that each agent performs the given exploration task independently without the knowledge of what other agents are doing. The collaboration means that the agents can cover the domain much faster than a single agent case and hence, effectively by communicating and sharing their information with other agents if they are within the communication range. Thus, the decentralized/collaborative exploration strategy can facilitate the multi-agent exploration while avoiding the issues from the single agent as well as the centralized multi-agent cases.

Given $n_a\in\mathbb{N}$ numbers of agents, the exploration planner must reflect each robot's energy level as it is finite. In the OT-based plan, the finite energy can be transformed into the total number of robot points, $M$, for each agent. For the given number of robot points $M\in\mathbb{N}$, all the points are equally weighted by $m(x_i^k)$, where $m(x_i^k) = \frac{1}{M}$, $\forall i$ and $x_i^k$ denotes the position of agent $k$ at discrete time $t\in\mathbb{N}_0$. Similarly, the given spatial distribution can be represented by $N\in\mathbb{N}$ numbers of sample points and each sample point $y_j$ has a uniform weight in the beginning, given by $n_0^k(y_j)=\frac{1}{N}$. Here, the weight $n_t^k(y_j)$ of a sample point $y_j$ decreases with time as the robot explores the domain, thus making the weight $n_t^k(y_j)$ a function of time $t$. Notice that the agents perform the exploration task in a decentralized manner and hence, the weight information for the sample points available to the agent $k$, $n_t^k(y_j)$, will differ from other agents' weight information. If one agent share the information on $n_t^k(y_j)$ with other agents within the communication range, then it will be unified by a certain rule, which will be explained later. The sample point positions $\{y_j\}_{j=1}^{N}$ are, however, assumed to be identical across all agents initially.


We consider that at the beginning of the exploration ($t=0$), the robot points for the agent $k$ are all accumulated at the initial robot position $x^k_0$. (There are a total of $M$ robot points, which are not determined yet initially.) As the robot updates its position from $x^k_0$ to $x^k_1$ in the next time step, the weight assigned for the new position $x^k_1$ becomes $m(x^k_1) = \frac{1}{M}$. All the remaining weights $\frac{M-1}{M}$ for future positions $\{x^k_i\}_{i=2}^{M}$, which are yet to be determined, are moved with the agent and these future positions are considered to be concentrated on the current robot position $x^k_1$. To generalize this description, the following is proposed.

\begin{assumption}\label{assump: remaining weight}
Given the number of robot points $M$ for the agent $k$, the weight for each point is given by $\frac{1}{M}$. 
For any time $t\in\mathbb{N}_0$, the past robot points $\{x_i^k\}_{i=1}^{t}$ possess a total of weights given by $\sum_{i=1}^{t}\left(\frac{1}{M}\right) = \frac{t}{M}$.
The undetermined future robot positions $\{x^k_i\}_{i=t+1}^{M}$ are all accumulated at the current robot position, $x_{t}$, which has remaining weights $\sum_{i=t+1}^{M}\left(\frac{1}{M}\right) = \frac{M-t}{M}$.
\end{assumption}

We provide the schematic of the proposed method to realize the efficient exploration in Fig. \ref{fig: layer}.
\begin{figure}[h!]
\begin{center}
\includegraphics[scale=0.4]{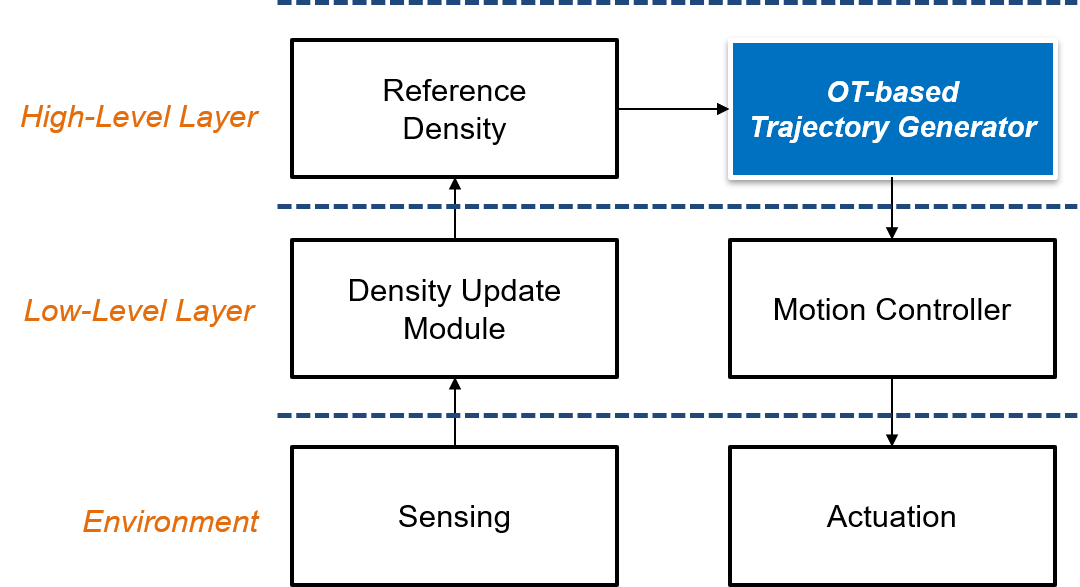}
\caption{Schematic of multiple layers with optimal transport (OT)-based trajectory generator placed on high-level layer}\label{fig: layer}
\end{center}
\end{figure}
The primary focus of this paper is to develop the OT-based trajectory generator in the high-level layer, which is the main contribution of this work. This goal is achieved by providing the information about reference distribution to the OT-based Trajectory Generator, which generates trajectories for each agent to follow. The motion controller in the low-level layer is decoupled from the trajectory generator, thus enabling the proposed method to be applicable to heterogeneous robot platforms. Since the developed method is not platform-specific, the efficient exploration can be achieved in collaboration between various robot such as unmanned aerial vehicles, ground robots, and unmanned underwater vehicles. As the agents explore areas of interest, they obtain data for an environment using on-board sensors. The density update module in the low-level layer receives the measured data and update the reference distribution (density) accordingly. 
Finally, the information about the reference distribution in the high-level layer will be updated through the density update module.

The optimal transport problem focuses on determining the non-negative optimal transport plan $\pi_{ij}^{k\star}$ for the given Euclidean distance $\Vert x_i^k - y_j\Vert$. Unlike conventional optimal transport problems in the LP form \eqref{eqn: LP}, the efficient, decentralized, and collaborative multi-agent exploration problem contains two parameters, $\pi^k_{ij}$ and $x_i^k$, both as the decision variables. 
This renders the efficient robot exploration problem much more difficult than the LP problem. In what follows, we introduce a two-stage approach to tackle this problem.


\subsection{Methodology: A Two-Stage Approach}

The developed method consists of two steps: the next goal point determination stage in a receding-horizon fashion and the weight update stage. 
To determine the next goal point for the agent $k$ to visit, an agent considers a fixed number of sample points within a certain range. Then, the agent computes a feasible future trajectory by connecting the sample points with non-negative weights within the range. The first sample point of that trajectory is considered as the next goal point and the agent moves towards that point using its motion controller. Once reached a new position, the agent updates weights of all sample points. In this weight update stage, the agent distributes $\frac{1}{M}$ of mass to the sample points with non-negative weight that are located nearby. The agent trajectory is governed by performing these two operations in every time step until weights of all sample points are completely depleted. More details about this process is provided below. 

\subsubsection{Next goal point ($^{g}x_{t+1}^k$) determination stage}

At any given discrete-time step $t$, if an agent $k$ is located at $x_t^k$, the next goal position ${^{g}x}_{t+1}^k$ for this agent can be computed by the following steps. The agent selects $h\in\mathbb{N}$ numbers of sample points $y_j$ by generating a circle with the center at the current robot position $x_t^k$ and an initial radius of $r_0$. The radius is incrementally increased by $\delta$ until the agent detects $h$ numbers of sample points within the circular search area. Once these points are found, all possible trajectories are generated by connecting all the sample points in the circle starting from $x_t$, as depicted in Fig. \ref{fig: next x} (a).

\begin{figure}[t!]
\begin{center}
\subfloat[]{\includegraphics[scale=0.2]{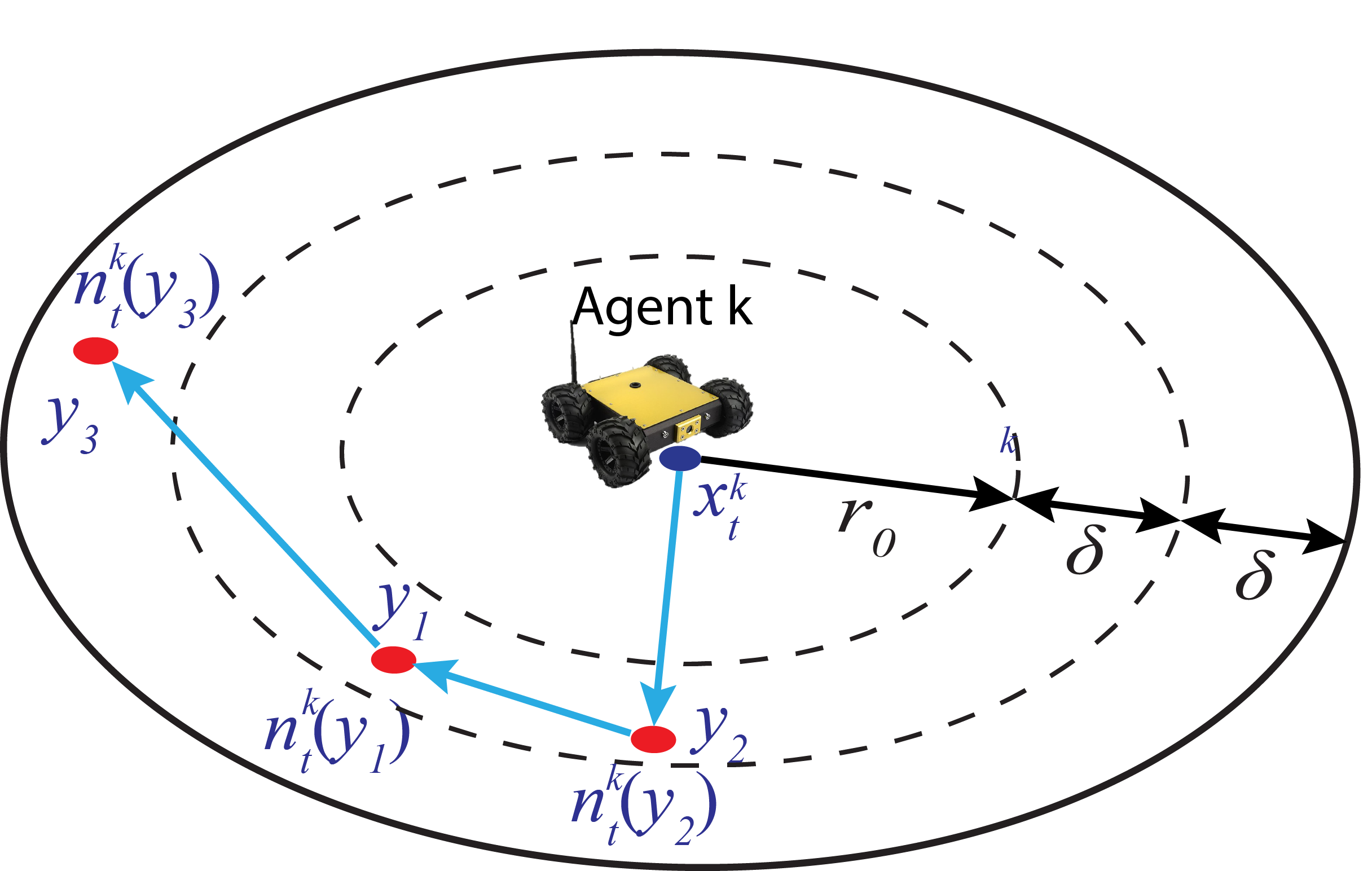}}\,\,
\subfloat[]{\includegraphics[scale=0.12]{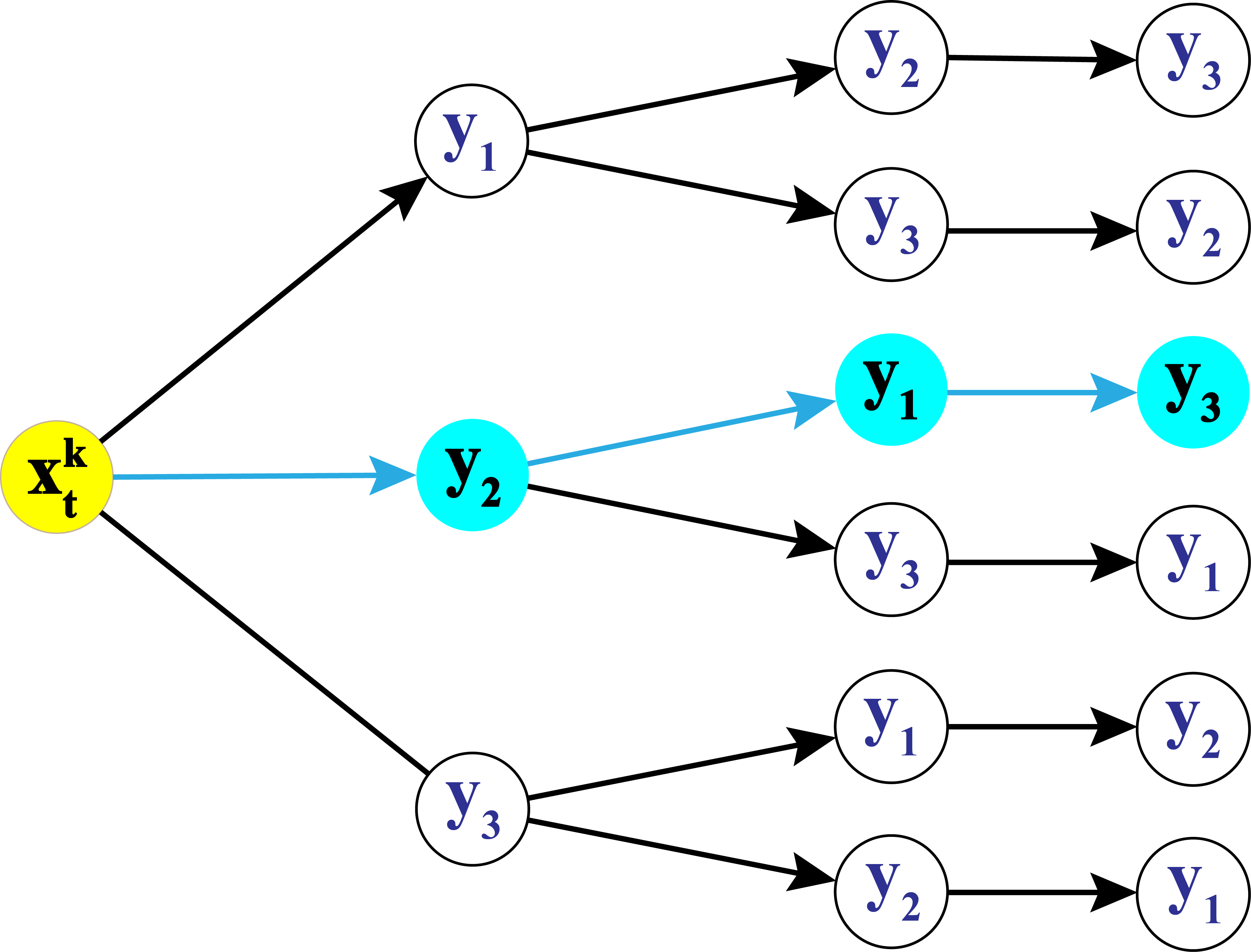}}
\end{center}
\caption{Schematic of the next goal point $^{g}x^k_{t+1}$ determination process: (a) increase the radius of the circle until $h$ numbers of points are found; (b) construct a tree associated with the detected points $y_j$}\label{fig: next x}
\end{figure}
For this purpose, a tree structure is constructed to connect all sample points in the circle starting from $x_t^k$. In this case, the size of the tree becomes ${h}!$, which reflects all possible trajectories as illustrated in Fig. \ref{fig: next x} (b), which has $h=3$. To calculate the sequence of sample points in the circle, a cost function is defined by
\begin{align}
    C^k(i) &= \dfrac{\lVert y_{\sigma_{t+1}}-x_{t}^k\rVert}{n^k_{t}(y_{\sigma_{t+1}})} + 
    \sum_{j=1}^{h-1}\dfrac{\lVert {y_{\sigma_{t+j+1}} - y_{\sigma_{t+j}}}\rVert}{n^k_{t}({y_{\sigma_{t+j+1}}})},\label{eqn: cost function}\\ 
    &\qquad\qquad\qquad\qquad i=1,2,\ldots,{h} !,\nonumber
\end{align}
where $y_{\sigma_{t+j}}$, $j=1,2,\ldots,{h}$, are the sample points located within the circle such that ${\sigma_{t+j-1}\neq\sigma_{t+j}}$, $\forall t\in\mathbb{N}_0$. 

The cost function $C^k(i)$ is defined by \eqref{eqn: cost function} in such a way that the agent follows a short trajectory. Also, the weight $n_t^k(y_j)$ for each sample point $y_j$ in the circle is reflected in \eqref{eqn: cost function} as we expect the agent to visit some points having large weights first.

Given the definition $x^k_{t+1:t+h}:=\{x^k_{t+1},x^k_{t+2},\ldots,x^k_{t+h}\}$, 
the candidate trajectory for the robot $^{c}x^k_{t+1:t+h}(i)$, $i=1,2,\ldots,h!$, is obtained from the tree construction. Then, the $h$-step optimal trajectory $^{g}x^k_{t+1:t+h}$ is determined by
\begin{align}
    ^{g}x^k_{t+1:t+h} = \{^{c}x^k_{t+1:t+h}({i^{\star})}\,\vert\, i^{\star}=\text{argmin}_{i} C^k(i)\}\label{eqn: x_{t+1}}
\end{align}
The agent considers the first point of $^{g}x^k_{t+1:t+h}$ as the next goal point, $^{g}x^k_{t+1}$, and approaches that point to visit. The robot may or may not be able to reach $^{g}x_{t+1}^k$ due to the robot motion constraints.


\subsubsection{Weight update stage}
After the agent $k$ has arrived at a new position $x^k_{t+1}$ (which, again, might be different from $^{g}x_{t+1}^k$), the weight information available to the agent $n^k_{t+1}(y_j)$ associated with each sample point $y_j$ is revised using the following weight update law:
\begin{align}
n_{t+1}^k(y_j) = n_{t}^k(y_j)  - [\pi_{(t+1)j}^{k}]^{\star},\, \forall j\label{eqn: weight update}
\end{align}
where $[\pi_{(t+1)j}^{k}]^{\star}$ is the optimal transport plan, which denotes the weight distribution scheme from $x^k_{t+1}$ to each $y_j$. This optimal transport plan $[\pi_{(t+1)j}^{k}]^{\star}$ is determined from the solution of the following LP problem:

\begin{equation}\label{eqn: LP for weight update}
  \begin{aligned}
    & \underset{\pi^k_{(t+1)j}}{\text{minimize}} & & \sum_{j}\pi^k_{(t+1)j}\Vert x^k_{t+1} - y_j \Vert\\
    & \text{subject to} & & \pi^k_{(t+1)j} \geq 0,\quad \sum_{j=1}^{N}\pi^k_{(t+1)j} = \dfrac{1}{M},\\
	& & & \pi^k_{(t+1)j} \leq \min\left( n^k_{t}(y_j), \frac{1}{M}\right), \, \forall j.
\end{aligned}
\end{equation}

The optimal solution calculated from the solution of LP problem \eqref{eqn: LP for weight update} quantifies how much of the weight $\frac{1}{M}$ for the agent position $x^k_{t+1}$ needs to be distributed to each sample point $y_j$. The first constraint in \eqref{eqn: LP for weight update} ensures that the transport plan $\pi_{(t+1)j}$ is non-negative. The second constraint is included for the law of mass conservation to indicate that the total weight distributed from $x_{t+1}$ to $y_j$ is equal to $\frac{1}{M}$. The last constraint is to guarantee that the transportation plan $\pi_{(t+1)j}$ cannot exceed the fixed predetermined capacity for each point. This constraint is attained by having the smaller value between the  distribution limit $\frac{1}{M}$ and receiving limit $n_{t}(y_j)$. 
After the determination of the optimal transport plan $[\pi_{(t+1)j}^{k}]^{\star}$, the weight for each sample point is updated by \eqref{eqn: weight update}.

The following proposition is developed for the analytic solution of \eqref{eqn: LP for weight update}.
\begin{proposition}\label{prop: analytic solution}
The optimal solution for the LP problem \eqref{eqn: LP for weight update} is obtained by repeating
\begin{align*}
\pi^k_{(t+1)j^{*}} &= 
\min\left(n^k_{t}(y_{j^{*}}), m(x^k_{t+1})\right),\\
\text{ where } 
j^{*} &= \argmin_{j\in\{j\vert n^k_{t}(y_j) > 0\}} \lVert x^k_{t+1} - y_j\rVert \\
m(x^k_{t+1}) &= m(x^k_{t+1}) - \pi^k_{(t+1)j^{*}}\\
n^k_{t}(y_{j^{*}}) &= n^k_{t}(y_{j^{*}}) - \pi^k_{(t+1)j^{*}}
\end{align*}
until $m(x_{t+1})$ becomes zero.
\end{proposition}

\begin{proof}
Given a single point $x^k_{t+1}$ in LP \eqref{eqn: LP for weight update}, the optimal transport plan for the agent is to deliver the maximum permissible weight to the closest points with positive weights in order, as long as the weight $m(x^k_{t+1})$ remains positive.  
\end{proof}

This two-stage strategy is repeated in a receding-horizon fashion, meaning that in every time step, the agents considers only $h$ numbers of sample points within the circular search area to determine where to go during the next goal point determination stage, followed by updating the weight of the sample points. Therefore, the parameter $h$ is given as the horizon length. As the agents cover the given domain and distributes the mass to the sample points, the weight of the sample points decreases and this process continues until all the weights of the sample points are completely depleted.

\subsection{Algorithm}

The formal algorithm for the decentralized exploration strategy is presented in Algorithm \ref{algorithm:1}. Initially, the starting positions of the agents $\{x_0^k\}_{k=1}^{n_a}$, the sample point representation for the spatial reference distribution $\{y_j\}$, the communication range $r_\text{comm}$, the number of robot points $M$, the horizon length $h$, initial search radius $r_0$ and the radius increment $\delta$ are given. At every time step, each agent generates a circle with a radius of $r$, where $r$ keeps increasing with an increment $\delta$ until it finds $h$ numbers of sample points $y_j$ having a positive weight (i.e., $\#\mathcal{R}(x^k_t,r) = h$ and $n^k_t(y_j) > 0$). Next, the agents compute all possible trajectories and corresponding costs $C^k(i)$ using \eqref{eqn: cost function}, determine the next goal point $^{g}x^k_{t+1}$ from \eqref{eqn: x_{t+1}}, and approach their corresponding new positions $x^k_{t+1}$ using a motion controller. After arriving at a new position, each agent distributes the weight to sample points and revises the weight information $n^k_t(y_j)$ from \eqref{eqn: weight update}. 

If an agent detects any other robot within the communication range, then the information sharing occurs between them. At any given time $t$, if an agent $k$ finds another agent $q$ within the communication range $r_{\text{comm.}}$ (i.e., distance $d_{kq} \leq r_{\text{comm.}} $), then the weight information for the sample points $y_j$, $j=1,2,\ldots, N$ is exchanged between agents and they update the weight information using the following rule: 
\begin{align}\label{eqn: decentralized information exchange}
   n^k_t(y_j) =  n^q_t(y_j) = \min(n^k_t(y_j), n^q_t(y_j)),\\\nonumber
   \qquad\qquad\qquad k, q \in \{ 1,2,\ldots,n_a\},  k\neq q
\end{align}
After the information exchange, each agent is able to grasp what sample points are already covered by other agents, leading to collaborative explorations by making them avoid areas already explored by others.

In the decentralized scheme, the total time for the exploration depends on the communication range as well as how frequently each agent communicates with others for information exchanges. If a communication range covers the entire domain, then this will enable all agents to communicate with other agents at every time step, which is technically the centralized exploration. In this scenario, the duration for the decentralized exploration is identical to that for the centralized exploration, computed by $\frac{M}{n_a}$.
On the other hand, given that no agent is able to exchange information with other agents due to the lack of communications, all agents will cover the domain independently, which is the same as the single agent case. In this case, the duration for the exploration is equal to the total number of robot points $M$. From this observation, it is evident that $\frac{M}{n_a}$ and $M$ are the lower and upper bounds of the actual exploration time for the decentralized strategy. An agent will continue its exploration until the weight information available to the agent $n_t^k(y_j)$ becomes zero for all sample points.

\begin{algorithm}[h!]
\caption{Decentralized Multi-Agent Exploration Algorithm}\label{algorithm:1}
\begin{algorithmic}[1]
\State initialize $x_0^k$, $y_j$, $M$, $N$, $r_0$, $r_{\text{comm.}}$, $\delta$, ${h}$, $n_a$, $t\gets 0$
\While{$n_t^k(y_j) >  0,\,  \forall j, \forall k $} 
\For{$k \gets 1$ to $n_a$} 
\If{$d_{kq} \leq r_{\text{comm.}}$}
\State update weight information from \eqref{eqn: decentralized information exchange}
\EndIf
\State initialize circle's radius by $r\gets r_0$
\While {$\#\mathcal{R}(x_t^k,r) \leq {h}$ and $n_{t}^k(y_j)>0$}
\State $r\gets r + \delta$
\EndWhile
\State calculate the cost function $C^k(i)$ associated with all possible candidate trajectories $^{c}x^k_{t+1:t+h}(i)$
\State obtain $^{g}x^k_{t+1}$ from \eqref{eqn: x_{t+1}}
\State update the robot position $x^k_{t}$ with the given robot motion controller and the goal position $^{g}x^k_{t+1}$
\State update weights $n_t^k(y_j)$ by \eqref{eqn: weight update}
\EndFor
\State $t\gets t+1$
\EndWhile
\end{algorithmic}
\end{algorithm}


\subsection{Performance Measure using Wasserstein Distance }
For large $M$ and $N$, it is difficult to compute the actual Wasserstein distance as this becomes computationally intractable. 
To measure the performance of the decentralized exploration scheme without any computational issues, the upper bound of the Wasserstein distance is developed, which can be utilized as a performance metric.

In the absence of an supervisory agent/computer, the agents do not have access to the weight information from other agents and hence, each agent needs to calculate its own performance. The set of neighboring agents within the communication range for the agent $k$ is denoted by $\mathcal{N}_k$. The Wasserstein distance for the agent $k$ is then computed by $ W^k(t) = \minimize_{\pi_{ij}^{k}}\sum_{k\in\mathcal{N}_k}\sum_{i=1}^{M}\sum_{j=1}^{N}\pi_{ij}^{k}\Vert x_{i}^{k} - y_j \Vert$ and the upper bound for this value is developed in the following theorem.
\begin{theorem}\label{thm: Decentral_W_UB}
Consider the optimization problem \eqref{eqn: LP} under Assumption \ref{assump: remaining weight} with robot points $\{x_i\}_{i=1}^{t}$ determined by the proposed efficient exploration algorithm. Then, at any time $t\in\mathbb{N}_0$, the Wasserstein distance $W^k(t)$ for the agent $k$ is upper bounded by
\begin{align}
    W^{k}(t) \leq \sum_{k\in\mathcal{N}_k}\sum_{i=1}^{t}\tilde{W}^k(i) + \sum_{k\in\mathcal{N}_k}\sum_{j=1}^{N}n_{t}^{k}(y_j) \cdot \lVert x_{t}^{k} - y_j\rVert,\label{eqn: Decentral_W_UB}
\end{align}
where $n_{t}^{k}(y_j)$ is the current weight for each $y_j$ after the weight update law \eqref{eqn: weight update} and $\tilde{W}^{k}(i):=\minimize_{\pi_{ij}^k}\sum_{j=1}^{N}\pi_{ij}^k\Vert x_i^{k} - y_j\Vert$ subject to the same constraints in \eqref{eqn: LP for weight update}.
\end{theorem}

\begin{proof}
At any time $t\in\mathbb{N}_0$, the current and previous robot points of agent $k$, $\{x^k_i\}_{i=1}^{t}$ as well as the remaining weights $n^k_t(y_j)$, $j=1,2,\ldots, N$, are given by the proposed algorithm. 
Under Assumption \ref{assump: remaining weight},  the future robot points are all accumulated at $x^k_{t}$.
Then, the Wasserstein distance at any time $t$ (constraints are omitted here) is upper bounded by
\begin{align*}
    W^k(t) &= \minimize_{\pi_{ij}^{k}}\sum_{k\in\mathcal{N}_k}\sum_{i=1}^{M}\sum_{j=1}^{N}\pi_{ij}^{k}\Vert x_{i}^{k} - y_j \Vert\\
    &\leq \minimize_{\pi_{ij}^{k}}\sum_{k\in\mathcal{N}_k}\sum_{i=1}^{t}\sum_{j=1}^{N}\pi_{ij}^{k}\Vert x_i^{k} - y_j\Vert + \\
    &\qquad\qquad\minimize_{\pi_{ij}^{k}}\sum_{k\in\mathcal{N}_k}\sum_{i=t+1}^{M}\sum_{j=1}^{N}\pi_{ij}^{k} \lVert x_{i}^{k} - y_j\rVert\\
    &\leq \sum_{k\in\mathcal{N}_k}\sum_{i=1}^{t}\underbrace{\left(\minimize_{\pi_{ij}^k}\sum_{j=1}^{N}\pi_{ij}^k\Vert x_i - y_j\Vert\right)}_{=\tilde{W}^{k}(i)} + \\
    &\qquad\qquad\sum_{k\in\mathcal{N}_k}\sum_{j=1}^{N}n_{t}^{k}(y_j) \cdot \lVert x_{t}^{k} - y_j\rVert,
\end{align*}
where the last inequality holds by Assumption \ref{assump: remaining weight} and the mass conservation law.
\end{proof}

To determine the upper bound of the Wasserstein distance at any time $t\in\mathbb{N}_0$ from \eqref{eqn: Decentral_W_UB}, it only requires computing $\tilde{W}^k(t)$, followed by the computation of the second term in \eqref{eqn: Decentral_W_UB} which is obtained from the weight update \eqref{eqn: weight update}.
$\tilde{W}(t)$ can be obtained analytically by Proposition \ref{prop: analytic solution} and the upper bound is computed recursively as the values for $\tilde{W}^k(i)$, $i=1,2,\ldots,t-1$, are already calculated, and thus known from the previous time step.
Therefore, this upper bound can be calculated by each agent without any computational issues, enabling real-time monitoring for the efficiency measure.

\section{\uppercase{Simulations}}
To validate the technical soundness of the proposed decentralized collaborative multi-agent exploration method, simulations are performed and the simulation results are presented in Fig. \ref{fig: decentralized}. For the simulation, the first order robot dynamics is considered, which is applicable to various robot platforms, such as ground mobile robots, multi-rotor UAVs, etc. The first order robot dynamics for continuous time is given as:
\begin{align}\label{eqn: first order continuous}
    \dot{x}(t) = \frac{dx(t)}{dt} = u(t)
\end{align}
where $x(t)\in\mathbb{R}^2$ is the continuous planar position of the agent and $u(t)\in\mathbb{R}^2$ is the instantaneous velocity as the control input for the first order dynamics.  

The counterpart of \eqref{eqn: first order continuous} for discrete-time case with the control input $u^k$ for an agent $k$ can be written as:
\begin{align}
    x^k_{t+1} = x^k_t + u^k\Delta t = x^k_t + u_{max}\frac{^gx^k_{t+1} - x^k_t}{||^gx^k_{t+1} - x^k_t||}\Delta t
\end{align}\label{eqn: robot motion control}
where $x^k_{t} = [\mathsf{x}^k_t, \mathsf{y}^k_t]^{T}$ is the agent position with $\mathsf{x}^k_t,\mathsf{y}^k_t\in\mathbb{R}$, $u_{max}$ is the maximum attainable speed of the agent, $\Delta t$ is the time interval for the discretization, and $^gx^k_{t+1}$ is the goal point for the next time step determined by \eqref{eqn: x_{t+1}} in the next goal point determination stage. 


In Fig. \ref{fig: decentralized}, the sample point representation of the given reference distribution is illustrated by the green dots. The spatial distribution considered for the simulation is given as a mixture of Gaussian with four modal Gaussian components as follows:
{\small
\begin{align*}
    \mu_1 &=[300, 600]^{T}, \quad \mu_2=[720, 275]^T, \\\quad \mu_3 &= [1300, 600]^T, \mu_4 =[1000, 1500]^{T},\\
    \Sigma_1 &= \begin{bmatrix}
       6000  &   0\\
          0  &   4500
    \end{bmatrix}, \quad
    \Sigma_2 = \begin{bmatrix}
        5250 &   0\\
           0 &   4750
    \end{bmatrix}, \\
    \Sigma_3 &= \begin{bmatrix}
        3250 &   0\\
           0 &   5000
    \end{bmatrix}, \quad
    \Sigma_4 = \begin{bmatrix}
        8000 &  0\\
           0 &  3500
    \end{bmatrix}
\end{align*}
}
Other simulation parameters are:
\begin{itemize}
  \item Domain size: $1500 \times 1800$
  \item Number of agents: $n_a = 3$
  \item Maximum allowable number of robot steps: $M=3000$ (for each agent)
  \item Number of sample points for the multi-modal Gaussian distribution: $N = 1600$
  \item Initial robot positions;\\
  $x_0 = [1000, 200]^{T}, [400, 900]^{T}, [1400, 400]^T$
  \item Maximum velocity of the robot: $10$
  \item Time interval for discretization: $\Delta t = 1$
  \item Robot communication range: $r_{\text{comm.}}=100$

\end{itemize}

\begin{figure*}[tbph!]
\centering
\subfloat[$t=0$]{\includegraphics[scale=0.36]{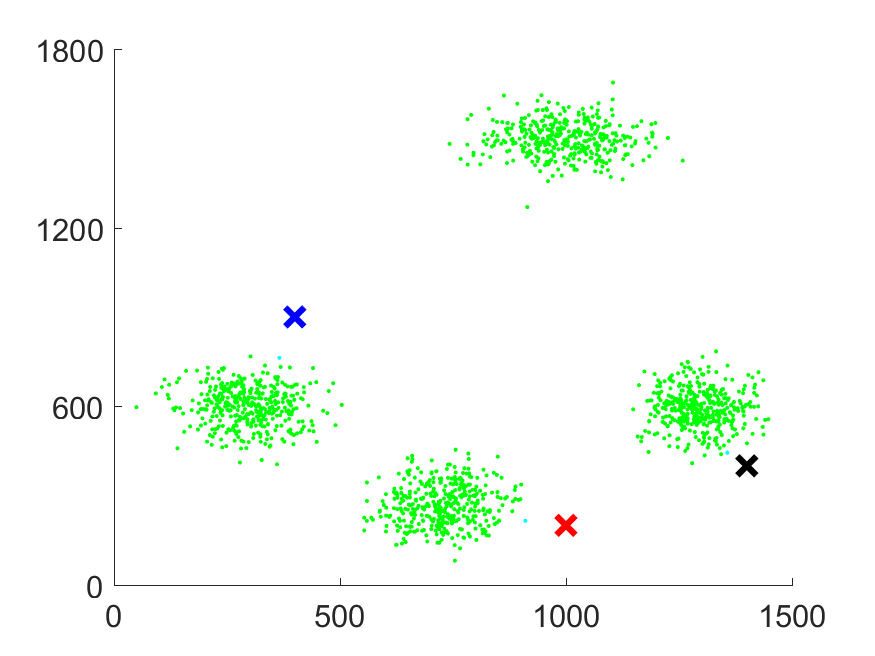}}
\subfloat[$t=600$]{\includegraphics[scale=0.36]{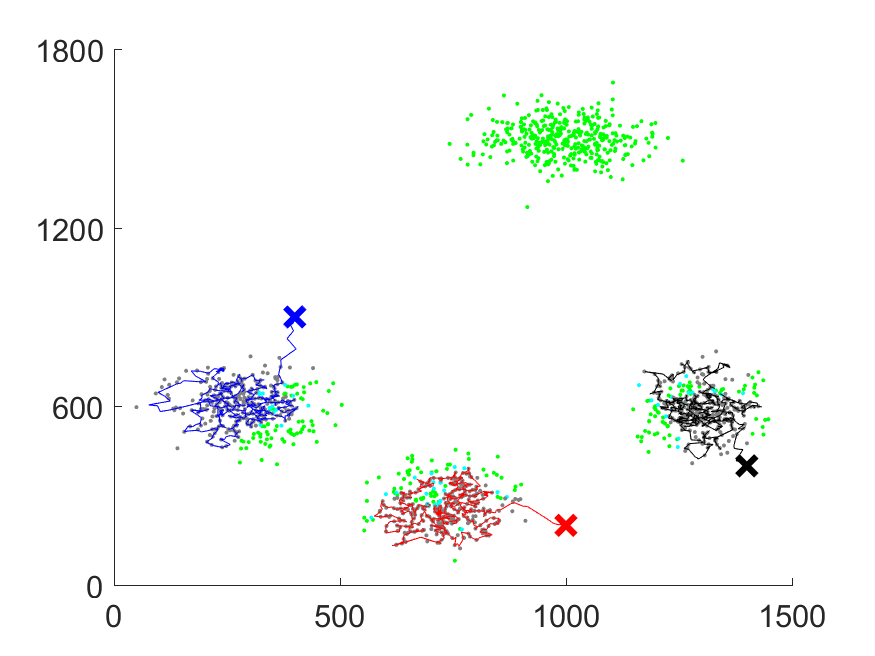}}
\subfloat[$t=800$]{\includegraphics[scale=0.36]{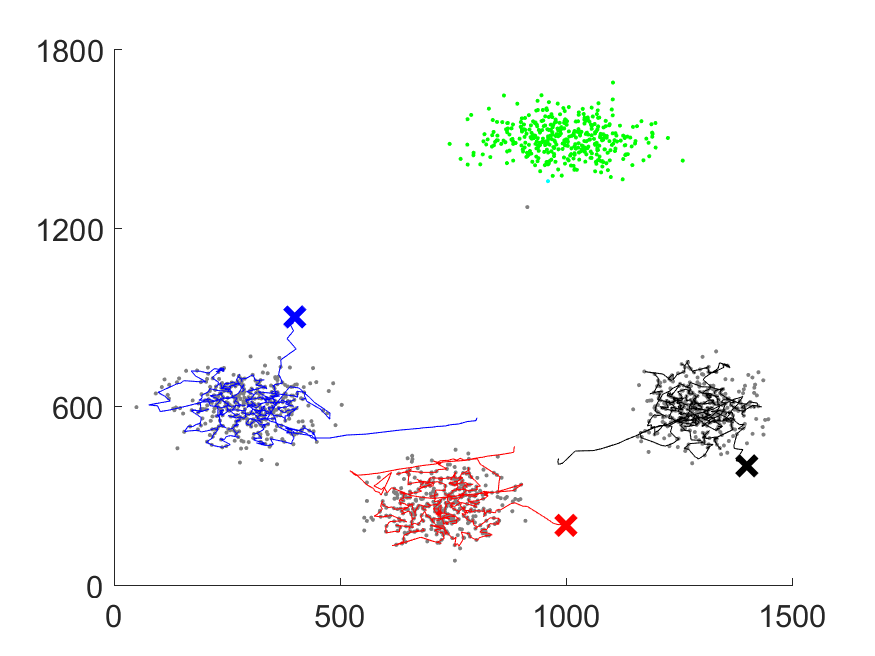}}\\
\subfloat[$t=1000$]{\includegraphics[scale=0.36]{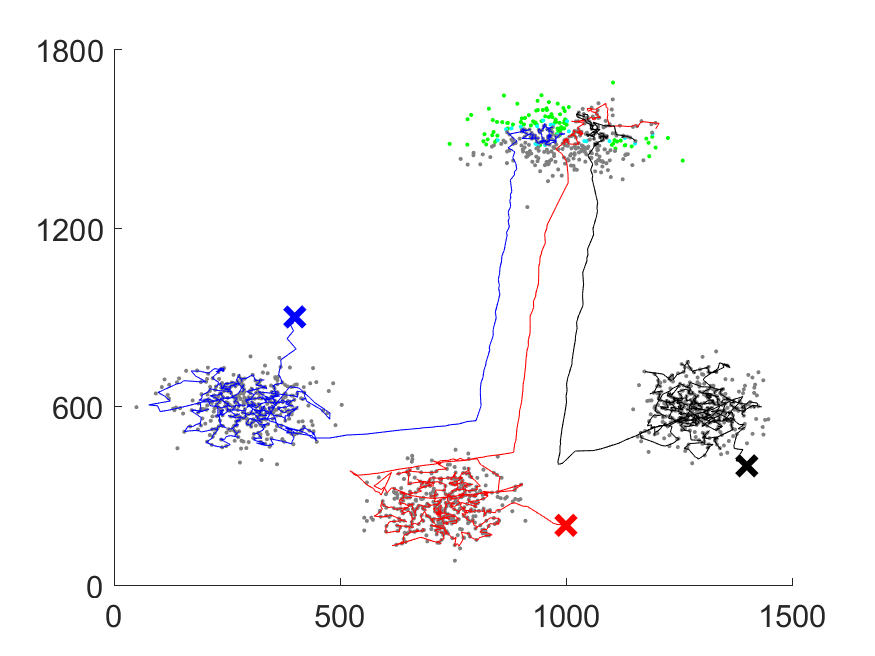}}
\subfloat[$t=1092$]{\includegraphics[scale=0.36]{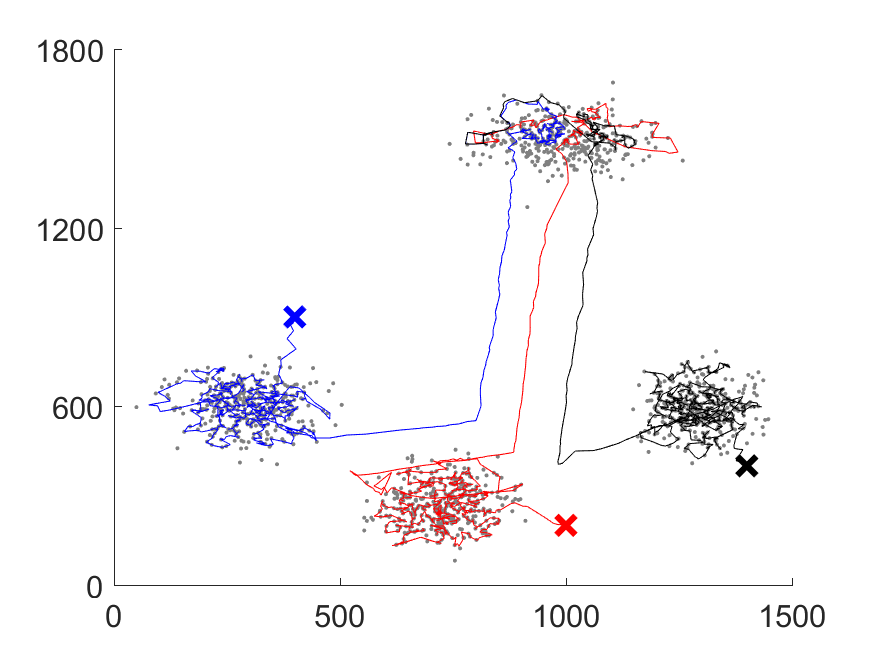}}
\subfloat[$W_{UB}$]{\includegraphics[scale=0.24]{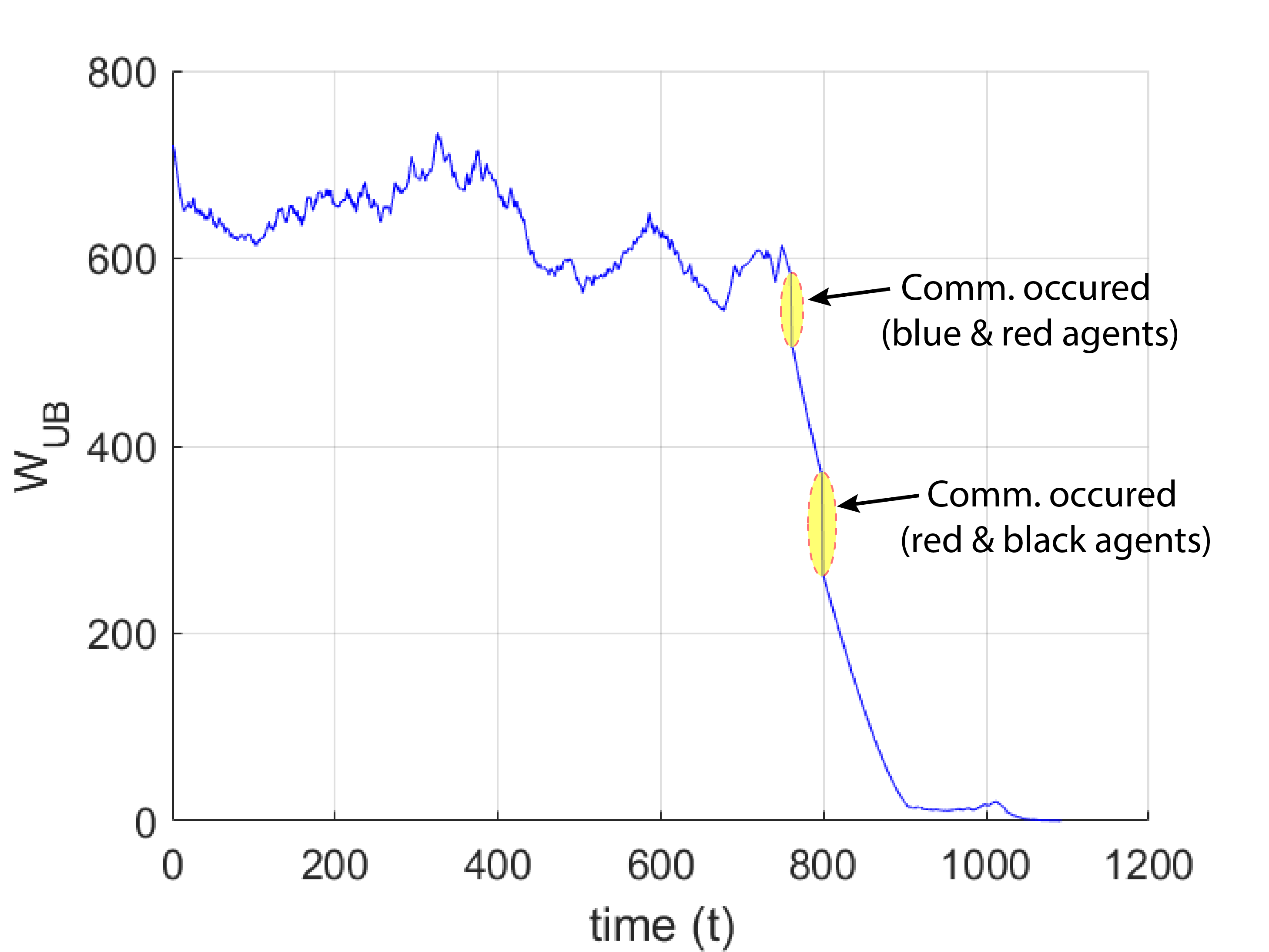}}
\caption{Snapshots of decentralized multi-agent exploration for the given spatial distribution (a) - (e); and the upper bound of the Wasserstein distance for the red agent (f)}\label{fig: decentralized}
\end{figure*} 

In Fig. \ref{fig: decentralized} (a), the initial positions of the agents are represented by blue, red and black crosses and the sample points with initial evenly distributed weight are shown in green dots. With time, the sample points lose weights due to the mass distribution by the agents. The sample points with depleted weight are shown as grey dots in later figures in Fig. \ref{fig: decentralized}. 

Initially, the agents are not in communication because of their starting positions, and therefore, they start the exploration task as completely independent agents. 
The left, middle and right distributions located at the lower part of the domain are the closest distributions from the initial position of the blue, red and the black agents, respectively. As a result, the agents approach and survey these regions separately. Fig. \ref{fig: decentralized} (b)-(c) depict the exploration of these areas by the agents following the two-stage approach.

Once the exploration of their respective areas is complete (Fig. \ref{fig: decentralized} (c)), each agent approaches the next closest unexplored region. 
The red agent aims at exploring the region that is covered the by the blue agent and starts moving towards it. At $t = 761$, the red and blue agents find themselves within the communication range, share their weight information, re-evaluate their decision on next area to visit, and approach that area. As they are not in communication with the black agent, they approach the lower right distribution for exploration and similarly, the black agent moves toward the lower mid distribution. At $ t = 799$, all three agents communicate (blue and black agents communicate through the red agent) and realize that all three distributions in the lower part of the domain are explored by one of them. As a result, the agents plan to visit the upper mid distribution. These events of information sharing and decision changing are captured in Fig. \ref{fig: decentralized} (c) and (d). 

During the exploration of the last distribution, the agents are located within the communication range most of the time, resulting in an efficient exploration by the agents as shown in Fig. \ref{fig: decentralized} (e), which illustrates that the trajectories of three agents do not overlap in most cases. Therefore, the agents work similarly to the centralized exploration scheme. 

The simulation termination time is set up as the largest time among each agent's time to completely deplete the weight of sample points. In this simulation, the exploration duration is $t = 1092$, which is greater than the lower limit for the exploration time $\frac{M}{n_a} = 1000$, but less than the maximum individual robot steps $M = 3000$. Therefore, it can be concluded that the three-agent decentralized collaborative system effectively explored the domain by reducing the time almost by one-third ($1092/3000$).

Fig. \ref{fig: decentralized} (f) provides the variation of the upper bound of the Wasserstein distance with time for the red agent as the agents cover the given domain. During the separate exploration of lower mid region by the agent, $W_{UB}$ decreases slowly. A sharp decrease of $W_{UB}$ is observed from $t = 749$ to $t=760$ due to the red agent's movement towards the lower left region. The upper bound $W_UB$ drops suddenly at $t = 761$ and $t=799$, when the red agent communicates with the blue and the black agents and exchange the weight information. From this point forward, $W_{UB}$ keeps decreasing and reaches a very small value of $0.1609$ at the final time step $t = 1092$. This quantified value using the Wasserstein distance implies that the multi-robot decentralized system is able to attain the efficient exploration. 


\section{\uppercase{Conclusion}}
This paper proposed an efficient, decentralized, and collaborative multi-agent exploration plan based on the OT theory to cover a domain associated with a given reference distribution. The reference distribution is represented by an ensemble and the agents perform an exploration mission in a receding-horizon manner, following a two stage approach. The information exchange occurs between any agents within the communication range (decentralized control), enabling them to make a decision efficiently and collaboratively based on the past coverage by other agents.
The upper bound of the Wasserstein distance was proposed as a metric to quantify the efficiency of the exploration plan in a computationally feasible way. The formal algorithm to realize the decentralized multi-agent exploration was provided. Finally, simulations were performed and results were presented to validate the proposed algorithm.



\bibliographystyle{ieeetr}
\bibliography{OT_multi_robot_exp_journal}

\end{document}